\documentclass[12pt,a4paper]{article}
\usepackage[dvipdfmx,hiresbb]{graphicx}
\usepackage[dvipdfmx]{color}

\usepackage{color}
\usepackage{bm,enumerate,amsmath,amssymb,amsthm}
\usepackage{epsfig}
\usepackage{cite}
\usepackage{algorithm}
\usepackage{algorithmic}
\usepackage{txfonts}
\usepackage{subfigmat}
\usepackage{figsize}
\usepackage{here}
\usepackage{listings}
\usepackage{braket}
\usepackage{comment}
\usepackage{lscape}
\usepackage{bigints}

\usepackage{arxiv}
\usepackage[utf8]{inputenc} 

\theoremstyle{definition}
\newtheorem{theorem}{Theorem}
\newtheorem*{theorem*}{Theorem}

\newtheorem*{definition*}{Definition}

\newtheorem{lemma}{Lemma}

\newtheorem{remark}{Remark}

\newtheorem{proposition}{Proposition}

\renewcommand{\headeright}{}

\title{Observable-Guided Generator Selection for Improving Trainability in Quantum Machine Learning with a $ \mathfrak{g} $-Purity Interpretation under Restricted Settings}

\author{
  Hiroshi Ohno \\
  Toyota Central R \& D Labs., Inc.\\
  Aichi, Japan \\
  \texttt{oono-h@mosk.tytlabs.co.jp} \\
}

\date{\empty}

\begin{document}
\maketitle

\begin{abstract}
  To study generator design for parameterized unitaries in quantum machine learning (QML), we propose an observable-guided generator selection algorithm for $ n $-qubit Pauli-string generator pools.
  The proposed method selects generators based on two criteria: maintaining large first-order sensitivity in the gradients and suppressing second-order interference in the Hessian matrix.
  Under a restricted setting with Pauli-string observables and candidate generators, the selection problem can be formulated as a binary optimization problem that favors mutually anti-commuting generators.
  Numerical experiments on a synthetic dataset with a small-scale five-qubit circuit show that the selected generators yield faster training than random generator selection in our setting, while exhibiting similar expressibility.
  Furthermore, under additional algebraic assumptions, the proposed criteria admit an interpretation in terms of the $ \mathfrak{g} $-purity of the observable: the first-order sensitivity is proportional to the $ \mathfrak{g} $-purity, whereas the second-order interference, namely the off-diagonal elements of the Hessian matrix, is upper-bounded by it.
  These results suggest that observable-guided generator selection is a promising direction for improving trainability in restricted QML settings.
\end{abstract}

\section{Introduction}
In this study, we propose an observable-guided generator selection algorithm for parameterized unitaries in quantum machine learning to accelerate training.
In the proposed algorithm, prior to training, we select generators $ G_{j} $ such that the first-order sensitivity, namely, $ \| \, [G_{j}, O] \, \|_{F}^{2} $, is kept large, while the second-order interference, $ \sum_{j \neq k} \| \, [G_{k}, [G_{j}, O]] \, \|_{F}^{2} $ in the Hessian matrix is suppressed.
Here, we assume that the generators are $ n $-qubit Pauli strings, and the norm $ \| \cdot \|_{F} $ denotes the Frobenius norm.
Keeping the off-diagonal elements of the Hessian matrix small suppresses strong couplings among variables and prevents the local curvature structure from becoming overly complicated, thereby reducing excessive coupling of search directions.
Hence, by suppressing this interference, the training speed is expected to improve.
Furthermore, under certain conditions, these metrics can be related to $ \mathfrak{g} $-purity \cite{ragone2024} through the Casimir operator.

Previous studies on generator design have exploited first-order sensitivity, expressibility, and metrics based on dynamical Lie algebra (DLA).
For example, in ADAPT-VQE \cite{grimsley2019}, during training, an operator is selected from the operator pool based on the magnitude of its gradient, and the ansatz is updated, gradually becoming problem-dependent.
In iterative qubit coupled cluster (iQCC) \cite{ryabinkin2020}, at each epoch, generator selection is performed based on gradient ordering, and the Hamiltonian is gradually updated.
To select generators efficiently, the direct interaction set has been introduced.
However, during training, the quantum circuit depth becomes large; thus, the likelihood of encountering barren plateaus increases.

Our contributions are summarized as follows:
\begin{itemize}
\item We propose a generator selection algorithm that explicitly controls the off-diagonal elements of the Hessian matrix, and we demonstrate its effectiveness via numerical experiments using small training datasets.
\item We reveal that the metrics used in the algorithm are connected to the $ \mathfrak{g} $-purity of the observable under certain conditions.
\end{itemize}

The remainder of the paper is organized as follows.
Related work is described in Section \ref{sec2}.
Section \ref{sec3} introduces the brief preliminaries necessary for this study, formulates a binary optimization problem, and presents the corresponding generator selection algorithm.
In Section \ref{sec4}, we describe the relationship between the metrics used in the proposed algorithm and the $ \mathfrak{g} $-purity of the observable.
Numerical experiments and results are presented in Section \ref{sec5}.
The training curves of parameterized unitaries constructed using the selected generators are compared with those obtained from randomly selected generators.
Finally, Section \ref{sec6} concludes the paper and discusses future work.

\section{Related work}\label{sec2}
\noindent
{\bf Adaptive training using a generator pool: }\\
In qubit-ADAPT-VQE \cite{tang2021}, a Pauli string pool is used instead of the fermionic operator pool employed in ADAPT-VQE, and the ansatz is iteratively grown in the same manner as in ADAPT-VQE \cite{grimsley2019}.
Qubit-ADAPT-VQE reduces circuit depth by an order of magnitude while maintaining comparable accuracy, thereby decreasing the likelihood of encountering barren plateaus.
The authors proposed a pool completeness criterion that determines whether a given pool can generate an exact ADAPT ansatz.
In addition, they demonstrated that the minimal number of pool operators satisfying this criterion increases linearly with the number of qubits.
The experimental results indicate the effectiveness and practicality of qubit-ADAPT-VQE on NISQ devices.

Qubit-ADAPT-VQE also exploits the first-order sensitivity, as in ADAPT-VQE.
Therefore, as adopted in our proposed algorithm, we believe that introducing the second-order interference into qubit-ADAPT-VQE can significantly accelerate the training process.

The iQCC with involutory linear combinations of Pauli products (QCC-ILCAP) \cite{lang2020} is a quantum algorithm for accurate ground-state electronic structure calculations on NISQ devices.
It employs a compact, adaptive ansatz built from anti-commuting, involutory linear combinations of Pauli products, reducing circuit depth and improving optimization convergence.
By iteratively selecting generators that satisfy $ A_{k}^{2} = I $, the method enables efficient implementation of $ e^{i \tau_{k} A_{k}} $ as simple rotations, enhancing numerical stability.
This shallow-circuit framework facilitates preparation of the optimal electronic wavefunction while reducing the risk of barren plateaus.

In QCC-ILCAP, operator ranking based on first-order sensitivity is employed.
Although the use of anti-commuting Pauli products enables shallow circuits and low computational cost, enhancing circuit expressibility may ultimately lead to an increase in circuit depth or the number of gates.
In iQCC and QCC-ILCAP, we believe that it is important, although challenging, to strike a balance between computational efficiency achieved through involutory operators and improved circuit expressibility achieved through iteration.
From this viewpoint, the use of a single operator composed of multiple anti-commuting Pauli products appears to be distinct from the second-order interference criterion adopted in our proposed algorithm.

\noindent
{\bf Symmetry-aware: }\\
In Ref. \cite{meyer2023}, the authors proposed a method called gate symmetrization, based on representation theory, in which a standard gate set is transformed into an equivalent one that respects the symmetry of the task.
Through this procedure, inductive bias is incorporated into the circuit, thereby improving generalization performance.
Experimental results on toy problems showed better generalization, demonstrating the effectiveness of symmetry-aware ansatz design.
In gate symmetrization, in order to construct a symmetry-preserving parameterized unitary, its generator is chosen so as to commute with the representation of the symmetry group.
For continuous groups, this condition can be formulated in terms of commutation with the generators of the group, namely, the Lie algebra.

To further improve generalization performance and training speed, we believe it is desirable to exploit the relationship between observables and generators, as well as information about the cost-function landscape, such as first-order sensitivity and second-order interference.

\noindent
{\bf Barren plateaus and $ \mathfrak{g} $-purity: }\\
In Ref. \cite{ragone2024}, the variance of the cost function for a parameterized unitary $ U(\theta) $ (a deep circuit) is given by
\begin{equation*}
  {\rm Var}_{\theta} [l_{\theta} (\rho, O)] = \frac{ \mathcal{P}_{\mathfrak{g}} (\rho) \, \mathcal{P}_{\mathfrak{g}} (O) }{ \dim(\mathfrak{g}) },
\end{equation*}
where $ \mathfrak{g} $ denotes the DLA of the generators of the unitary and is simple, $ l $ denotes the cost function, $ \rho \in i \mathfrak{g} $ denotes input state, and $ O \in i \mathfrak{g} $ denotes an observable.
When the dimension of the DLA is large (high expressibility), the $ \mathfrak{g} $-purity of the input state is small (strong generalized entanglement), or the $ \mathfrak{g} $-purity of the observable is small (low generalized locality), the variance becomes small, thereby increasing the likelihood of barren plateaus.
In the derivation of the above equation, the evaluation of the gradient variance in deep circuits corresponds to two-body correlators in the two-copy space.
In the Haar average, which acts as a projection operator with respect to the Lie group action, the split quadratic Casimir is used as the simplest invariant that captures this invariant component.

\section{Method}\label{sec3}
In this section, we describe the training model considered in this study, namely, a parameterized unitary generated by generators, and propose a generator selection algorithm.

\subsection{A training model and Hessian matrix}
We consider the following model:
\begin{equation}\label{eq1}
  U(\theta) = \prod_{l=1}^{L} e^{-i \theta_{l} G_{l}},
\end{equation}
where $ G_{l} $ denote a generator, $ iG_{l} \in \mathfrak{su}(d) $, $ \theta_{l} $ denote a trainable parameter, and $ L $ denotes the circuit depth.
Next, we consider an expectation-value-based cost function.
Then, a cost function $ C(\theta) $ is defined as
\begin{equation}\label{eq2}
  C(\theta) = \braket{\psi \, | \, U^{\dagger}(\theta) O U(\theta) \, | \, \psi},
\end{equation}
where $ \ket{\psi} $ denotes an initial state and $ O $ denotes an observable.

Here, as a motivating special case, for the generators, assuming that $ [G_{j}, G_{k}] = 0 $ for all $ j $ and $ k $, we obtain the following straightforward results.
Then, the gradient and the second derivative with respect to the parameter $ \theta $ are given by
\begin{equation}\label{eq3}
  \frac{\partial C(\theta)}{\partial \theta_{j}} = i \braket{\psi \, | \, U^{\dagger}(\theta) \, [G_{j}, O] \, U(\theta) \, | \, \psi},
\end{equation}
\begin{equation}\label{eq4}
  \frac{\partial^{2} C(\theta)}{\partial \theta_{j} \theta_{k}} = -\braket{\psi \, | \, U^{\dagger}(\theta) \, [G_{k}, [G_{j}, O]] \, U(\theta) \, | \, \psi}.
\end{equation}
These expressions motivate the use of $ [G_{j},O] $ as a first-order sensitivity indicator and $ [G_{k},[G_{j},O]] $ as a second-order interference indicator.
In the general non-commuting case, these expressions are no longer exact Hessian elements of the product ansatz; nevertheless, their norms remain useful local algebraic measures of how the observable responds to infinitesimal generator actions.

From the viewpoint of avoiding barren plateaus, it is preferable for $ [G_{j}, O] $ in Eq. \ref{eq3} to be large.
To accelerate training, interactions among parameters should ideally vanish.
Therefore, in Eq. \ref{eq4}, $ [G_{k}, [G_{j}, O]] $ should ideally vanish.
Based on these considerations, we propose an algorithm for selecting generators that satisfy these conditions (criteria).

In this study, we restrict the generators to $ n $-qubit Pauli strings to simplify the derivation of the algorithm.

\subsection{Binary optimization problem}
Before formulating the generator selection problem, we emphasize that the following criterion should be understood as a local algebraic design metric rather than an exact expression for the Hessian of a general non-commuting product ansatz.
For a product unitary $ U(\theta) = \prod_{l=1}^{L} e^{-i\theta_{l} G_{l}} $, when the generators do not mutually commute, the exact gradient and Hessian generally involve parameter-dependent conjugations of the generators by the surrounding circuit layers.
In contrast, the quantities $ [G_{j}, O] $ and $ [G_{k}, [G_{j}, O]] $ considered below characterize the local response of the observable under infinitesimal transformations generated by $ G_{j} $ and $ G_{k} $.
Therefore, we use these commutator norms as observable-level surrogate metrics for first-order sensitivity and second-order interference.
Under the restricted Pauli-string setting described below, this local observable-algebra criterion leads to a simple binary optimization problem favoring mutually anti-commuting generators.

First, to construct a generator selection algorithm, we present the following proposition to formally describe the optimization problem.
\begin{proposition}\label{pro1}
  Let $ \mathcal{P}_{n} $ be the set of $ n $-qubit Pauli strings.
  Assume that for the observable $ O $, $ \{G_{j}, O\} = 0 $ and $ \{G_{k}, O\} = 0 $ with $ G_{j}, G_{k} \in \mathcal{P}_{n} $.
  Then,
  \begin{equation}
    \| \, [G_{k}, [G_{j}, O]] \, \|_{F}^{2} = \left \{
    \begin{array}{ll}
      0 & (\{ G_{k}, G_{j} \} = 0)\\
      2^{n+4} & ([G_{k}, G_{j}] = 0).
    \end{array}
    \right.
  \end{equation}
\end{proposition}
\begin{proof}
  Since $ \{G_{j}, O\} = 0 $, we have $ [G_{j}, O] = 2 G_{j} O $.
  Then, $ [G_{k}, [G_{j}, O]] = 2 (G_{k} G_{j} O - G_{j} O G_{k}) $.
  When $ [G_{k}, G_{j}] = 0 $, $ [G_{k}, [G_{j}, O]] = 2 G_{j} [G_{k}, O] = 4 G_{j} G_{k} O$.
  Since $ G_{j} $, $ G_{k} $, and $ O $ are $ n $-qubit Pauli strings, $ \| \, 4 G_{j} G_{k} O \, \|_{F}^{2} = 2^{4} (\sqrt{2^{n}})^{2} $.
  Therefore, $ \| \, [G_{k}, [G_{j}, O]] \, \|_{F}^{2} = 2^{n+4} $.
  Next, when $ \{G_{k}, G_{j}\} = 0 $, $ [G_{k}, [G_{j}, O]] = 2 G_{j} [G_{k}, O] = -2 G_{j} \{G_{k}, O\}$.
  Therefore, $ \| \, [G_{k}, [G_{j}, O]] \, \|_{F}^{2} = 0 $.
\end{proof}
Proposition \ref{pro1} implies that, under the present assumptions, the off-diagonal elements of the Hessian matrix vanish whenever the selected generators mutually anticommute.
Therefore, maximizing the number of anticommuting generator pairs provides a principled way to suppress second-order interference.

\begin{remark}
The formulation in Proposition \ref{pro1} relies on restrictive assumptions, namely, a Pauli-string observable and a Pauli-string generator pool with $ \{ G, O \} = 0 $.
Therefore, the resulting binary optimization problem should be understood as a specialized formulation for this setting.
Generalization to broader classes of observables and generators remains an important direction for future work.
\end{remark}

Next, let the candidate set be $ S \coloneqq \{ G \in \mathcal{P}_{n} \; | \; \{G,O\} = 0 \} $.
From Proposition \ref{pro1}, for $ G_{j} $, $ G_{k} \in S $, the binary optimization problem is formulated as
\begin{equation}
  \max_{x_{j} \in \{0, 1\}} \sum_{j < k} c_{jk} x_{j} x_{k} \; {\rm subject} \; {\rm to} \sum_{j} x_{j} = L.
\end{equation}
where
\begin{equation}
    c_{jk} = \left \{
    \begin{array}{ll}
      1 & (\{ G_{k}, G_{j} \} = 0)\\
      0 & ([G_{k}, G_{j}] = 0).
    \end{array}
    \right.
\end{equation}
Note that this optimization problem may be NP-hard.
Thus, we use a brute-force algorithm to solve it because, in our numerical experiments, the number of qubits is set to five, which is relatively small.
For larger qubit systems, genetic algorithms, which are a type of metaheuristic algorithm, may be useful in practice.

\section{Relationship to $ \mathfrak{g} $-purity of an observable}\label{sec4}
In this section, we investigate the relationship between the gradient or Hessian matrix and the $ \mathfrak{g} $-purity of an observable \cite{ragone2024}.
In the elements of Hessian matrix of Eq. \ref{eq4}, we observe a double commutator such as $ [G_{k}, [G_{j}, O]] $.
Therefore, in order to relate this to the $ \mathfrak{g} $-purity of an observable, we need to assume a condition under which this double commutator can be regarded as a Casimir operator.

First, we state the following theorem on the gradient.
\begin{theorem}\label{th1}
  Assume that the generators $ i G_{j} \in \mathfrak{su}(d) $, and that $ G_{j} $ is an orthonormal basis of traceless Hermitian generators normalized as $ {\rm Tr}(G_{j} G_{k}) = \delta_{jk} $.
  Let the observable $ O \in {\rm span}(G_{j}) $.
  Then, we have
  \begin{equation}
    \sum_{j = 1}^{d^{2}-1} \| \, [G_{j},O] \, \|_{F}^{2} = c \| \, O \, \|_{F}^{2},
  \end{equation}
  where $ c $ denotes a constant.
\end{theorem}
\begin{proof}
  \begin{equation}
    \| \, [G_{j}, O] \, \|_{F}^{2} = {\rm Tr}( [G_{j}, O]^{\dagger} [G_{j}, O] ) = - {\rm Tr}( [G_{j}, O]^{2} ).
  \end{equation}
  Then we have
  \begin{equation}\label{eq-th11}
    - {\rm Tr}( [G_{j}, O]^{2} ) = {\rm Tr}(-G_{j} O G_{j} O) + {\rm Tr}(O G_{j} G_{j} O) + {\rm Tr}(G_{j} O O G_{j}) + {\rm Tr}(-O G_{j} O G_{j}).
  \end{equation}
  Next, we expand $ {\rm Tr}( O [G_{j}, [G_{j}, O]] ) $ as follows:
  \begin{equation}\label{eq-th12}
    {\rm Tr}( O [G_{j}, [G_{j}, O]] ) = {\rm Tr}(-O G_{j} O G_{j}) + {\rm Tr}(O G_{j} G_{j} O) + {\rm Tr}(-G_{j} O G_{j} O) + {\rm Tr}(O O G_{j} G_{j}).
  \end{equation}
  Using the cyclic property of the trace, Eq. \ref{eq-th11} is equal to Eq. \ref{eq-th12}; therefore, $ - {\rm Tr}( [G_{j}, O]^{2} ) = {\rm Tr}( O [G_{j}, [G_{j}, O]] ) $.
  We consider $ \sum_{j} {\rm Tr}(O [G_{j}, [G_{j}, O]]) $.
  Here, $ \mathcal{C}_{{\rm ad}} \coloneqq \sum_{j} ad_{G_{j}}^{2} $, where $ ad_{G_{j}} (X) \coloneqq [G_{j}, X] $.
  $ i G_{j} \in \mathfrak{su}(d) $, $ G_{j} $ is an orthonormal basis, and $ O \in {\rm span}(G_{j}) $.
  Since $ \mathfrak{su} $ is simple, and the adjoint representation of $ \mathfrak{su}(d) $ is irreducible, $ \mathcal{C}_{\rm ad} $ acts as $ c I $ on the adjoint representation, where $ c $ is the eigenvalue of the quadratic Casimir operator in the adjoint representation of $ \mathfrak{su}(d) $.
  Hence, we have $ \sum_{j} [G_{j}, [G_{j}, O]] = c O $.
  Therefore, $ \sum_{j} {\rm Tr}(O [G_{j}, [G_{j}, O]]) = {\rm Tr}(O c O) = c {\rm Tr}(O^{2}) = c \| \, O \, \|_{F}^{2} $.
\end{proof}

Here, we present two useful lemmas related to Theorem \ref{th1}.
\begin{lemma}\label{lem1}
  Under the same assumption as in Theorem \ref{th1}, we have
  \begin{equation}
    \sum_{j,k} \| \, [G_{k}, [G_{j}, O]] \, \|_{F}^{2} = c^{2} \| \, O \, \|_{F}^{2}.
  \end{equation}
\end{lemma}
\begin{proof}
  \begin{equation}
    \begin{split}
      \| \, [G_{k}, [G_{j}, O]] \, \|_{F}^{2} &= \| ad_{G_{k}}(ad_{G_{j}}(O)) \|_{F}^{2} = \| ad_{G_{k}}ad_{G_{j}}(O) \|_{F}^{2}\\
      &= \braket{ (ad_{G_{k}}ad_{G_{j}}O)^{\dagger}, ad_{G_{k}} ad_{G_{j}}O }_{F}\\
      &= \braket{ ad_{G_{k}}ad_{G_{j}}O, ad_{G_{k}} ad_{G_{j}}O }_{F}\\
      &= \braket{ ad_{G_{j}}O, ad_{G_{k}} ad_{G_{k}} ad_{G_{j}}O }_{F}\\
      &= \braket{ O, ad_{G_{j}} ad_{G_{k}}^{2} ad_{G_{j}}O }_{F}.
    \end{split}
  \end{equation}

  Here, $ \braket{X, Y}_{F} = {\rm Tr}(X^{\dagger} Y) $ and $ \braket{ ad_{G_{k}}X, Y}_{F} = \braket{ X, ad_{G_{k}}Y}_{F} $.
  Referring to the proof of Theorem \ref{th1}, since $ \sum_{k} ad_{G_{k}}^{2} = c I $.
  Therefore, we have
  \begin{equation}
    \begin{split}
      \sum_{j, k} \braket{ O, ad_{G_{j}} ad_{G_{k}}^{2} ad_{G_{j}} O }_{F} &= \sum_{j} \left< O, ad_{G_{j}} \sum_{k} ad_{G_{k}}^{2} ad_{G_{j}} O \right>_{F}\\
      &= \sum_{j} \braket{ O, c \, ad_{G_{j}} ad_{G_{j}}O }_{F}\\
      &= \left< O, c \sum_{j} ad_{G_{j}}^{2} O \right>_{F}\\
      &= c^{2} \braket{ O, O }_{F}\\
      &= c^{2} \| \, O \, \|_{F}^{2}.
    \end{split}
  \end{equation}
\end{proof}

\begin{lemma}\label{lem2}
  Under the same assumption as in Theorem \ref{th1}, we have
  \begin{equation}
    \sum_{j = 1}^{d^{2}-1} \| \, [G_{j}, [G_{j}, O]] \, \|_{F}^{2} \geq \frac{c^{2}}{d^{2} - 1} \| \, O \, \|_{F}^{2}.
  \end{equation}
\end{lemma}
\begin{proof}
  Using the Hilbert-Schmidt inequality, $ \sum_{j} \| \, A_{j} \, \|_{F}^{2} \geq \frac{1}{d^{2}-1} \| \, \sum_{j} A_{j} \, \|_{F}^{2} $.
  Here, $ A_{j} = [G_{j}, [G_{j}, O]] $.
  Referring to the proof of Theorem \ref{th1}, $ \sum_{k} ad_{G_{k}}^{2} = c I $.
  Therefore, we have
  \begin{equation}
    \begin{split}
      \sum_{j = 1}^{d^{2}-1} \| \, [G_{j}, [G_{j}, O]] \, \|_{F}^{2} &\geq \frac{1}{d^{2} - 1} \| \, \sum_{j} [G_{j}, [G_{j}, O]] \, \|_{F}^{2}\\
      &= \frac{c^{2}}{d^{2} - 1} \| \, O \, \|_{F}^{2}.
    \end{split}
  \end{equation}
\end{proof}

Finally, using Lemmas \ref{lem1} and \ref{lem2}, we obtain the following theorem concerning the Hessian matrix.
\begin{theorem}\label{th2}
  Under the same assumption as in Theorem \ref{th1}, we have
  \begin{equation}
    \sum_{j \neq k} \| \, [G_{k}, [G_{j}, O]] \, \|_{F}^{2} \leq c^{2} \frac{d^{2}-2}{d^{2}-1} \| \, O \, \|_{F}^{2}.
  \end{equation}
\end{theorem}
\begin{proof}
  From Theorem \ref{th1},
  \begin{equation}
    \begin{split}
      \sum_{j, k} \| \, [G_{k}, [G_{j}, O]] \, \|_{F}^{2} &= \sum_{j} \| \, [G_{j}, [G_{j}, O]] \, \|_{F}^{2} + \sum_{j \neq k} \| \, [G_{k}, [G_{j}, O]] \, \|_{F}^{2}\\
      &= c^{2} \| \, O \|_{F}^{2}.
    \end{split}
  \end{equation}
  Therefore, from Lemmas \ref{lem1} and \ref{lem2},
    \begin{equation}
      \begin{split}
        \sum_{j \neq k} \| \, [G_{k}, [G_{j}, O]] \, \|_{F}^{2} &\leq c^{2} \| \, O \|_{F}^{2} - \frac{c^{2}}{d^{2}-1} \| \, O \|_{F}^{2}\\
        &= c^{2} \frac{d^{2}-2}{d^{2}-1} \| \, O \, \|_{F}^{2}.
      \end{split}
    \end{equation}
\end{proof}

Next, we consider the relation to the $ \mathfrak{g} $-purity of the observable $ O $.
The DLA $ \mathfrak{g} $, spanned by the generators, determines the reachable Lie group $ G $ ($ \exp(i \mathfrak{g}) $) of the quantum circuit.
Under the assumption of Theorem \ref{th1} and assuming $ O \in \mathfrak{g} $, the $ \mathfrak{g} $-purity of the observable $ O $ is defined as $ \mathcal{P}_{\mathfrak{g}}(O) \coloneqq \| \, O \, \|_{F}^{2} $.
Here, $ \{G_{j}\} $ is an orthonormal basis of the DLA.
Roughly speaking, the $ \mathfrak{g} $-purity of an observable measures how well the observable matches the DLA.
Therefore, the sum of the gradient elements is proportional to $ \mathcal{P}_{\mathfrak{g}}(O) $, and the sum of the off-diagonal elements of the Hessian matrix is upper-bounded by $ \mathcal{P}_{\mathfrak{g}}(O) $.
Theorem \ref{th1} implies that when $ \mathcal{P}_{\mathfrak{g}}(O) $ is small, the gradient becomes small, thereby increasing the possibility of barren plateaus.
In contrast, Theorem \ref{th2} implies that even if $ \mathcal{P}_{\mathfrak{g}}(O) $ is large, the sum of the off-diagonal elements can still become small, thereby facilitating training.

\section{Numerical experiments and results}\label{sec5}
The aim of the numerical experiments is to demonstrate the improvement in training speed achieved by the proposed algorithm.
We compare the training curves obtained using generators selected by the algorithm with those obtained using randomly selected generators.

\noindent
{\bf Training data preparation: }
To prepare the training data, we used a quantum circuit with five qubits and depth $ L = 5 $.
The circuit parameters $ \theta $ were uniformly randomly selected from the range $ [-\pi, \pi] $.
The input data were uniformly randomly selected from the range $ [0, 2\pi] $.
Angle encoding with $ R_{Y} $ gates was employed.
Five generators were randomly selected from the set of 5-qubit Pauli strings.
The measurements were given by the expectation value of the observable $ O = Z \otimes I^{\otimes 4} $.
The number of training samples was set to 100 to ensure a reasonable wall time.

\noindent
{\bf Generator generation results: }
Table \ref{tab2-1} summarizes the results for randomly selected generators (Random), generators selected by the proposed algorithm (Algorithm), generators selected using the metric $ \# [G_{j}, O] = 0 $ alone ($ [G, O] $ only), and generators selected using the metric $ \# [G_{j}, G_{k}] = 0 \, (j < k) $ alone ($ [G, G] $ only).
The table also reports the Hellinger distance between the output-state distribution of each circuit and the Haar-random distribution.
Here, the Hellinger distance is used as a measure of circuit expressibility for two reasons: (1) it is bounded between 0 and 1, and (2) its coefficient of variation is empirically smaller than that of the KL divergence.
To estimate the distance distribution, we used 500 samples and 50 bins.
\begin{table}[htb]
  \centering
  \caption{Evaluation results of randomly selected generators (Random) and generators selected by the algorithm (Algorithm) and each metric. Generator selection was performed 20 times using different random seeds. The table reports the mean values, with one standard deviation shown in parentheses.}\label{tab2-1}
  \vspace{5pt}
  \begin{tabular}{ccccc}\hline
    Metric & Random & Algorithm & $ [G, O] $ only & $ [G, G] $ only\\ \hline
    $ \# [G_{j}, O] = 0 $ & 2.35 (1.35) & 0 (0) & 0 (0) & 2 (0)\\
    $ \# [G_{j}, G_{k}] = 0 $ ($ j < k $) & 5.05 (1.43) & 0 (0) & 5 (0)& 0 (0)\\
    Hellinger distance & 0.322 (0.0375) & 0.289 (0.0117) & 0.290 (0.0105) & 0.323 (0)\\ \hline
  \end{tabular}
\end{table}
For Algorithm, both $ \# [G_{j}, O] = 0 $ and $ \# [G_{j}, G_{k}] = 0 \, (j < k) $ were zero for all selected generators.
This indicates that the gradients remain sufficiently large while the off-diagonal terms vanish.
For Random, approximately half of the generators satisfied $ \# [G_{j}, O] = 0 $, whereas about half of the $ \# [G_{j}, G_{k}] = 0 \, (j < k) $ terms were zero.
In addition, the Hellinger distances were similar, suggesting that the two types of circuits have comparable expressibility.
The results for each method were consistent with the corresponding constraint it imposes.

\noindent
{\bf Training experiment setup: }
The quantum circuits (Eq. \ref{eq1})\footnote{
  The Rademacher complexity of this training model is $ \mathcal{O} \left( \frac{L}{\sqrt{M}} \right) $, where $ M $ denotes the number of training data. 
} were implemented in PennyLane \cite{bergholm2018}.
For the input data, angle encoding with $ R_{Y} $ gates was also employed.
Parameter optimization was performed using the simultaneous perturbation stochastic approximation (SPSA) \cite{gacon2021} with a momentum term of 0.5 and a learning rate of 0.001.
The initial parameters were uniformly randomly selected from the range $ [-0.1, 0.1] $.
The number of epochs was 200.
The training trials for circuits using randomly selected generators and generators selected by the proposed algorithm were each performed 20 times with different random seeds.

\noindent
{\bf Training curves: }
Figure \ref{fig3-2} shows the training curves for Random, Algorithm, and the two single-metric variants.
The average RMSE values and one standard deviation, shown as error bars, are presented.
In the figure, the RMSE in each trial was normalized by its value at epoch 0, and the mean and standard deviation were then computed across trials.
\begin{figure}[htbp]
  \centering
  \begin{minipage}{13.5cm}
    \SetFigLayout{1}{2}
    \subfigure[Algorithm and Random]{\includegraphics[width=6cm]{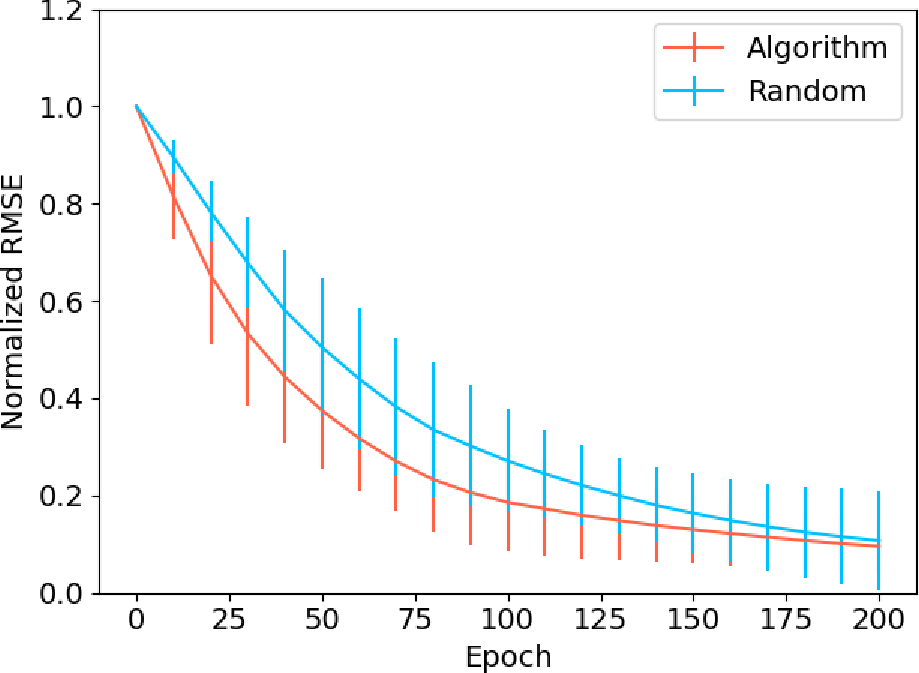}}
    \hfill
    \subfigure[Algorithm and each metric]{\includegraphics[width=6cm]{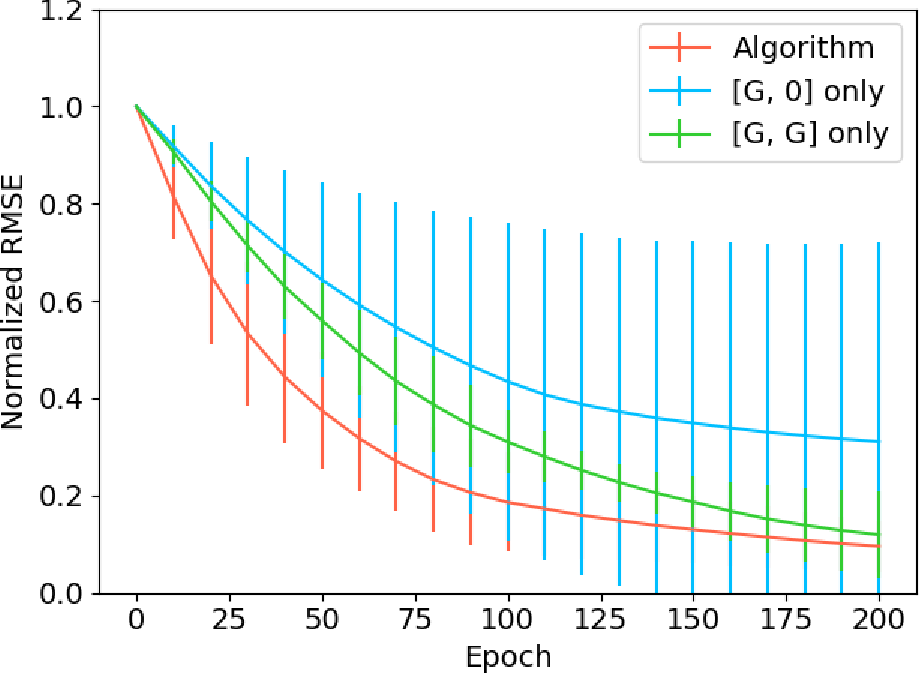}}
  \end{minipage}
  \caption{Training curves of randomly selected generators (Random) and generators selected by the algorithm (Algorithm) and each metric. Training trials of the quantum circuits using generators selected by Random and Algorithm were each performed 20 times with different random seeds. The average RMSE values are shown, and the error bars indicate one standard deviation.}\label{fig3-2}
\end{figure}
The curve for Algorithm lies below that for Random over most of the training process, although the two curves gradually converge.
This suggests that the proposed algorithm accelerates early-stage training.
These results support the effectiveness of the algorithm, particularly during the early stage of training.
At epoch 200, the original RMSE values of the two methods become closer, possibly because the two circuits have comparable expressibility.
A two-sided Student’s t-test yielded $ p = 0.063 $, indicating that the difference at this stage is not statistically significant at the 5\% level.

Comparison with the single-metric results suggests an advantage of combining the two metrics in Algorithm.
However, the variances observed for the single-metric methods were relatively large.
One possible reason is the limited size of the training dataset.
In these experiments, the number of training samples was chosen to keep the training wall time within a reasonable range.

\section{Conclusion}\label{sec6}
In this study, we investigated generator design for parameterized unitaries in QML and proposed an observable-guided generator selection algorithm to accelerate training.
The proposed algorithm selects generators from the set of $ n $-qubit Pauli strings such that the first-order sensitivity remains large while second-order interference is suppressed.
We note that the algorithm may be NP-hard.
Numerical experiments on a synthetic dataset of 100 samples using a five-qubit quantum circuit demonstrated that the proposed algorithm yields faster training than random generator selection.
Furthermore, under certain restricted conditions, we showed that the selection criterion used in the algorithm is related to the $ \mathfrak{g} $-purity of the observable \cite{ragone2024}.
Consequently, the relationship between the gradients or Hessian matrix and the $ \mathfrak{g} $-purity of the observable was clarified.

Although the present numerical results are encouraging, they are limited to small-scale settings.
Further numerical studies will be necessary to more comprehensively validate the effectiveness of the proposed algorithm.

Finally, future work will focus on the following directions:
\begin{itemize}
\item We will investigate how the formulation can be extended beyond Pauli-string observables and anti-commuting Pauli-string generator pools.
\item For large-qubit quantum circuits, we will develop computationally efficient selection algorithms based on genetic algorithms.
\item We will incorporate the proposed algorithm into adaptive training methods such as ADAPT-VQE \cite{grimsley2019,tang2021} and iQCC \cite{lang2020}, and evaluate its effectiveness.
\item We will also incorporate additional criteria, such as symmetry \cite{meyer2023} and expressibility, to further improve training speed and accuracy.
\end{itemize}

\bibliographystyle{plain}
\bibliography{references}

\end{document}